%% file: main.tex
\documentclass[journal]{IEEEtran}
\usepackage{amsmath,amsfonts}
\usepackage{algorithmic}
\usepackage{algorithm}
\usepackage{array}
\usepackage{textcomp}
\usepackage{stfloats}
\usepackage{url}
\usepackage{verbatim}
\usepackage{graphicx}
\usepackage{cite}
\usepackage{amssymb}
\usepackage{xurl}
\usepackage{xcolor}

\usepackage{caption}  
\usepackage{subcaption}
\usepackage{soul}

\usepackage{amsthm}
\newtheorem{thm}{Theorem}[section]

\newtheorem{prop}[thm]{Proposition}
\newtheorem{rem}{Remark}

\begin{document}

\title{Partitioned Mixed Small Gain--Phase Decentralized Stability Criterion for Power Systems}

\author{Diego Cifelli, Adolfo Anta
\thanks{D.Cifelli and A.Anta are with AIT Austrian Institute of Technology, Vienna, Austria, email: \{diego.cifelli, adolfo.anta\}@ait.ac.at. The authors disclose the use of Claude (Anthropic) and Gemini (Google) in the preparation of this manuscript, for language editing, for drafting text from author-provided input, and for assistance in the derivation of the analytical results. The authors verified all results and take full responsibility for the content of this article.}
}



\maketitle

\begin{abstract}
The increasing penetration of converter-interfaced resources is making power-system stability assessment more challenging, particularly in heterogeneous grids containing both grid-forming and grid-following converters. Existing decentralized mixed small-gain and small-phase criteria provide scalable stability certificates, but they require all converters to satisfy the same type of condition at a given frequency. As a result, they cannot simultaneously exploit the low gain of grid-following converters and the favorable phase properties of grid-forming converters, leading to unnecessary conservatism. This paper proposes a partitioned mixed gain–phase decentralized stability criterion that allows distinct subsets of converters to satisfy different local requirements at the same frequency. Specifically, one subset can be certified through small-gain bounds, while the complementary subset is certified through small-phase bounds. The admissible trade-off between gain and phase margins is determined by a network-dependent quadratic constraint, yielding a technology-aware stability certificate that remains local at the converter level. The paper characterizes the admissible set of gain and phase bounds, establishes useful convexity and boundedness properties, and develops a practical procedure for selecting these bounds. The proposed method is demonstrated on heterogeneous systems containing grid-forming and grid-following converters, including a two-converter system and the IEEE 39-bus system, outperforming the standard decentralized small-gain and small-phase conditions.
\end{abstract}


\input{sec1_Introduction}
\input{sec2_Stability}
\input{sec3_DecConditions}

\input{sec4_ConverterWise}
\input{sec5_ParSelection}
\input{sec6_StudyCases}
\input{secN_Conclusion}

\bibliographystyle{IEEEtran}
\bibliography{DecCondition_bib}

\appendices

\section{Proof of Propositions}
\subsection{Proof of Proposition \ref{prop:gamma_conv}}
\label{prof:gamma_conv}
Fix $\kappa$, and $\theta$, and denote by $P(\gamma)$ the left-hand side of \eqref{eq:Jnet_QC_dec}, i.e.,
\begin{equation}
    P(\gamma)=
    -\gamma \underbrace{\begin{bmatrix}
        [\mathbf{H}_{net}^{gg}]^* \\ [\mathbf{H}_{net}^{g\phi}]^*
    \end{bmatrix}
    \begin{bmatrix}
        \mathbf{H}_{net}^{gg} & \mathbf{H}_{net}^{g\phi}
    \end{bmatrix}}_{A}+
    \frac{1}{\gamma}
    \underbrace{\begin{bmatrix}
        I_m & 0 \\
        0 & 0
    \end{bmatrix}}_{B}+C
\end{equation}
Since $(\tilde{\gamma},\theta) \in \mathcal{L}_\kappa$, then $P(\tilde{\gamma}) \succeq 0$.
Then, $P(\gamma)$ can be rewritten as $P(\gamma) = P(\tilde{\gamma})+(P(\gamma)-P(\tilde{\gamma}))$, so $P(\gamma) \succeq 0$ if $P(\gamma)-P(\tilde{\gamma}) \succeq 0$. By direct computation we obtain:
\begin{equation}
    P(\gamma)-P(\tilde{\gamma}) = -(\gamma-\tilde{\gamma})A+\left(\frac{1}{\gamma}-\frac{1}{\tilde{\gamma}}\right)B
\end{equation}
Since $\gamma \leq \tilde{\gamma}$ and $A,B \succeq 0$, then $P(\gamma)-P(\tilde{\gamma}) \succeq 0$. Therefore, $P(\gamma)\succeq 0$ and $(\gamma,\theta) \in \mathcal{L}_\kappa$.
\subsection{Proof of Proposition \ref{prop:theta_conv}}
\label{prof:theta_conv}
Let $\gamma$ and $\kappa$ be fixed, and write \eqref{eq:Jnet_QC_dec} as
\begin{equation*}
    \begin{bmatrix}
        P & Q + \kappa e^{j\theta} [\mathbf{H}_{net}^{\phi g}]^* \\
        Q^* + \kappa e^{-j\theta} \mathbf{H}_{net}^{\phi g} &
        R + \kappa e^{j\theta}[\mathbf{H}_{net}^{\phi\phi}]^*
          + \kappa e^{-j\theta}\mathbf{H}_{net}^{\phi\phi} 
    \end{bmatrix} \succeq 0
\end{equation*}
with $P=\gamma^{-1}I_m-\gamma[\mathbf{H}_{net}^{gg}]^*\mathbf{H}_{net}^{gg}$, $Q=-\gamma[\mathbf{H}_{net}^{gg}]^*\mathbf{H}_{net}^{g\phi}$, and $R=-\gamma[\mathbf{H}_{net}^{g\phi}]^*\mathbf{H}_{net}^{g\phi}$. We assume $P \succ 0$, since otherwise $\mathcal{L}_\kappa(\omega)$ is empty. By the Schur complement, \eqref{eq:Jnet_QC_dec} is equivalent to
\begin{equation*}
    S(\theta) = A_0 + \kappa e^{-j\theta} B + \kappa e^{j\theta} B^* \succeq 0,
\end{equation*}
with $A_0 = R - Q^*P^{-1}Q - \kappa^2\, \mathbf{H}_{net}^{\phi g} P^{-1} [\mathbf{H}_{net}^{\phi g}]^*$ and $B = \mathbf{H}_{net}^{\phi\phi} - \mathbf{H}_{net}^{\phi g} P^{-1} Q$, where $A_0 \preceq 0$. For any unit vector $x$, the condition $x^*S(\theta)x \geq 0$ reads
$
    2\kappa \operatorname{Re}\!\left(e^{-j\theta}\, x^*Bx\right)
    \;\geq\; -x^*A_0x \;\geq\; 0,
$
which confines $\theta$ to an arc of length at most $\pi$ centered at $\angle(x^*Bx)$. The set of admissible $\theta$ is therefore an intersection of arcs of length at most $\pi$. The intersection is itself an arc of length at most $\pi$, and the admissible values of $\theta$ form an interval, which proves the claim.
\subsection{Proof of Proposition \ref{prop:conditions_with_bounds}}
\label{prof:conditions_with_bounds}
    We first show that the existence of such a $\theta$ implies \eqref{eq:thA:2_1}--\eqref{eq:thA:2_3}. Since $\theta \ge \underline{\theta}$, \eqref{eq:thA:2_1} follows from \eqref{eq:phase_upper_bound}; similarly, \eqref{eq:thA:2_2} follows from \eqref{eq:phase_lower_bound} and $\theta \le \overline{\theta}$. Isolating $\theta$ in \eqref{eq:phase_upper_bound} and \eqref{eq:phase_lower_bound} for the converters attaining $\max_{i\in\mathcal{I}_\phi}\overline{\phi}(J_{C_i}(j\omega))$ and $\min_{i\in\mathcal{I}_\phi}\underline{\phi}(J_{C_i}(j\omega))$ yields $L<\theta<U$, with $L:=-\frac{\pi}{2}-\min_{i\in\mathcal{I}_\phi}\underline{\phi}(J_{C_i}(j\omega))$ and $U:=\frac{\pi}{2}-\max_{i\in\mathcal{I}_\phi}\overline{\phi}(J_{C_i}(j\omega))$, and \eqref{eq:thA:2_3} follows from $L<U$.

Conversely, assume \eqref{eq:thA:2_1}--\eqref{eq:thA:2_3} hold. A $\theta$ satisfies \eqref{eq:phase_upper_bound}--\eqref{eq:phase_lower_bound} if and only if $\theta \in (L,U)$; it thus suffices to show $[\underline{\theta},\overline{\theta}] \cap (L,U) \neq \emptyset$. Condition \eqref{eq:thA:2_3} guarantees $L<U$, while \eqref{eq:thA:2_1} and \eqref{eq:thA:2_2} give $U>\underline{\theta}$ and $L<\overline{\theta}$, respectively; hence the intersection is non-empty. Finally, by Proposition~\ref{prop:theta_conv}, every $\theta\in[\underline{\theta},\overline{\theta}]$ satisfies $(\gamma,\theta)\in\mathcal{L}_\kappa(\omega)$, so any $\theta$ in the intersection is admissible.

\end{document}

%% file: sec1_Introduction.tex
\section{Introduction}

The increasing penetration of converter-interfaced resources is changing the dynamic behavior of power systems and has already led to unexpected oscillatory phenomena in practical installations \cite{Cigre2024}. As a result, stability assessment methods for converter-dominated grids must scale to large numbers of devices, heterogeneous technologies, and changing operating conditions. Classical small-signal stability approaches, such as eigenvalue analysis or generalized Nyquist criteria, either require white-box models, do not scale well in large grids or become difficult to interpret for large systems. 
In addition, many system-level studies aggregate the dynamics of each component and provide only a binary stable/unstable answer. For practical grid integration, this is not sufficient: a useful criterion should also provide intuition on the units responsible for a possible instability, and what type of controller modification is needed to recover stability.

Among other factors, these requirements have motivated the development of
decentralized frequency-domain stability criteria. The objective is
to certify the stability of a multi-converter system through local
conditions on each converter, together with a network-dependent
condition. Such criteria are attractive because they can be used
modularly: a converter behaviour can be assessed without requiring full access to the internal models of
all other devices. Moreover, they can provide valuable insights for control design, e.g., they can indicate
whether a converter should reduce its gain in the low-frequency range.

A well-established decentralized certificate is passivity. If the network and all converters are passive, stability follows
independently of the detailed interconnection. This makes
passivity appealing, and related requirements are starting to
appear in grid-code specifications for grid-forming converters
\cite{fingridGridCode}. However, passivity is also conservative: at low frequencies, in the
synchronous $dq$ frame, the constant-power behavior of converters
introduces fundamental non-passive characteristics that cannot be eliminated
simply by tuning inner loops. For this reason, recent work has
proposed less restrictive alternatives, including passivity indices,
mixed gain--phase criteria, scaled relative graph
methods, DW shells and parametric
certificates for specific grid-forming control structures
\cite{chenExtendedFrequencyDomainPassivity2025a,chenUnifiedFlexibleFrequencyDomain2025,huangGainPhaseDecentralized2024,
baron2025decentralized,woolcockMixedGainPhase2023,deyPassivityBasedDecentralizedCriteria2023,zhang2025phantom,häberle2025decentralizedparametricstabilitycertificates}.


All these approaches impose identical requirements for grid-following and grid-forming converters. For instance, one of the limitations of the mixed gain\&phase criterion is that, at each frequency, all units are required to satisfy the same type of condition. This uniform requirement is not easy to justify in realistic grids composed of both grid-forming (GFM) and  grid-following (GFL) units. Indeed, GFM converters often exhibit large gains at low frequencies, hence it is more likely that they satisfy small-phase conditions. Conversely, GFL converters may fulfill gain requirements at low frequencies while failing to satisfy appropriate phase requirements. Hence, a fully decentralized condition that ignores the characteristics of each unit can become unnecessarily conservative or inconclusive, even when the actual interconnected system is stable.

This paper addresses this limitation by decoupling the conditions in different sets using limited information on converter technology.
We focus on mixed small-gain and small-phase conditions, as they provide a particularly useful compromise between scalability, complexity and design insight. A converter-wise partitioned mixed gain--phase criterion is proposed where, instead of imposing a uniform gain or phase requirement on all converters over the same frequency range, different local conditions are assigned to different converter subsets. At each frequency, one subset may be certified through a small-gain bound, while the complementary subset is certified through a small-phase bound. These bounds are coupled by a network-dependent quadratic constraint, which characterizes the admissible trade-off between gain margins for one subset and phase margins for the other.

The resulting conditions can be interpreted as technology-aware decentralized stability certificates. They remain decentralized at the converter level, since each converter is still checked through a local admittance condition, without the need for the models of other converters. However, they are less conservative than fully uniform decentralized tests because they exploit the fact that specific converters are connected at specific buses and belong to specific technology classes. 
Still, in the spirit of decentralized conditions, the proposed approach does not only determine whether the system is stable: it also indicates which converters fail their assigned bounds and whether the required corrective action is gain reduction, phase reshaping, or a different allocation of converter-wise requirements.
As such, the developed results can be also useful towards the definition of grid-code requirements for inverter-based resources: the conditions herein presented can be tested locally while ensuring system-level stability.

The contributions of this paper are as follows. First, we leverage the connections between quadratic constraints and small-gain and small-phase criteria to extend existing criteria by introducing a third partitioned option, in which some converters satisfy gain bounds and others satisfy phase bounds at the same frequency. This modular condition is amenable to other decentralized conditions that can be represented by means of quadratic constraints. Second, we define the set of admissible gain and phase bounds generated by the network for a prescribed converter partition, and characterize its main properties. Third, we develop a systematic procedure to select practical gain and phase bounds from this admissible set, avoiding an exhaustive search over all possible choices. Finally, the method is illustrated on heterogeneous systems with both GFM and GFL converters, showing that the proposed converter-wise criterion can certify stability in cases where standard decentralized small-gain and small-phase conditions are inconclusive. Data and code to reproduce the results are available at \cite{Repo}.

%% file: sec2_Stability.tex
\section{Mixed Gain--Phase Stability Conditions}
\label{sec:gain_phase_crit}

For a matrix $A\in \mathbb{C}^{n\times n}$, the gains are its singular values, with $\overline{\sigma}(A)$ and $\underline{\sigma}(A)$ the largest and smallest, respectively. The phases are defined through the numerical range $W(A)=\{ x^*Ax:\,x\in\mathbb{C}^n,\, \|x\|=1\}$, a convex subset of $\mathbb{C}$ containing the eigenvalues of $A$. The matrix $A$ is sectorial if $W(A)$ lies within an angular sector and $0 \notin W(A)$, thus it admits a factorization $A=T^*DT$ with $D$ diagonal, and the phases of $A$ are the arguments of the entries of $D$. Here $\overline{\phi}(A)$ and $\underline{\phi}(A)$ denote the maximum and minimum phases, with $\overline{\phi}(A)-\underline{\phi}(A) < \pi$, and $\Phi(A)=[\underline{\phi}(A),\overline{\phi}(A)]$. Geometrically, these correspond to the angular boundaries of $W(A)$. This definition extends to quasi-sectorial (i.e. $0 \in \partial W(A)$) and to semi-sectorial (i.e. $0 \in \partial W(A)$ and $\overline{\phi}(A)-\underline{\phi}(A) \leq \pi$) matrices \cite{chenPhaseTheoryMIMO2022}.

\subsection{Mixed Gain--Phase Stability Criterion}
\begin{figure}
    \centering
    \begin{subfigure}{0.45\columnwidth}
        \centering
        \includegraphics[width=\linewidth]{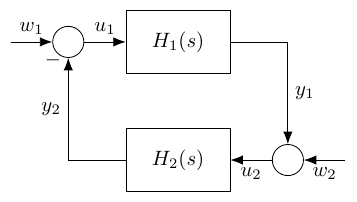}
        \caption{Generic Feedback System}
        \label{fig:fd_diagram}
    \end{subfigure}
    \begin{subfigure}{0.45\columnwidth}
        \centering
        \includegraphics[width=\linewidth]{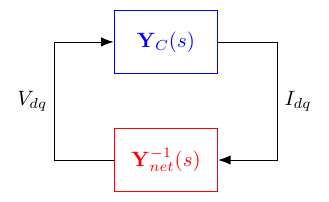}
        \caption{Admittance Feedback}
        \label{fig:fd_diagram_conv_net}
    \end{subfigure}
    \caption{Diagram of the feedback connection.}
    \label{fig:twoside}
\end{figure}
The concepts introduced above provide a framework for assessing the stability of linear time-invariant systems. Consider the feedback interconnection of $H_1(s)$ and $H_2(s)$ shown in Figure \ref{fig:fd_diagram}. 
The following theorem gives a sufficient condition for the stability of the resulting closed-loop system \cite{zhao2022smallgainmeetssmall}.
\begin{thm}[Mixed Small Gain--Phase Theorem]
    Let the open-loop systems $H_1, H_2$ be real, rational, stable, and proper transfer function matrices. Then, the closed-loop system is stable if for each $\omega \in [0, \infty]$, either
    \begin{enumerate}
        \item The small-gain condition holds:
        \begin{equation}
            \overline{\sigma}(H_1(j\omega))\overline{\sigma}(H_2(j\omega))<1
            \label{eq:gen_gain_cond}
        \end{equation}
        \item 
        or $H_1(j\omega)$ is semi-sectorial and $H_2(j\omega)$ is sectorial, and the small-phase condition holds:
        \begin{equation}
            \begin{aligned}
            &\overline{\phi}(H_1(j\omega))+\overline{\phi}(H_2(j\omega))<\pi \\
            &\underline{\phi}(H_1(j\omega))+\underline{\phi}(H_2(j\omega))>-\pi
            \end{aligned}
            \label{eq:gen_phase_cond}
        \end{equation}
    \end{enumerate}
    \label{thm:mixed_small_gain_phase}
\end{thm}
This theorem allows the stability of the closed-loop system to be assessed by verifying, at each frequency, either a gain or a phase condition on $H_1$ and $H_2$ separately. Both conditions are decoupled, in the sense that they are expressed in terms of the individual gains and phases of $H_1$ and $H_2$, without requiring the explicit computation of the closed-loop transfer function.

\color{black}
\subsection{Stability Conditions as Quadratic Constraints}
Theorem \ref{thm:mixed_small_gain_phase} can be restated in a more generic form in terms of Quadratic Constraints (QC) as follows \cite{ringh2025gain}.

\begin{thm} \label{thm:IQC_theorem}
    Let the open-loop systems $H_1$ and $H_2$ be real, rational, stable, and proper transfer function matrices. The closed-loop system is stable if there exists a Hermitian matrix 
    \begin{equation}
        \Pi(\omega) = \begin{bmatrix}
            \Pi_{11}(\omega) & \Pi_{12}(\omega)\\
            \Pi_{12}^*(\omega) & \Pi_{22}(\omega)
        \end{bmatrix}
    \end{equation}
    such that for all $\omega \in [0, \infty]$, $\Pi_{11}(\omega) \preceq 0$, $\Pi_{22}(\omega) \succeq 0$, and the following quadratic constraints are satisfied:
    \begin{subequations} \label{eq:qc}
        \begin{align}
            \begin{bmatrix}
                H_1(j\omega)\\
                I
            \end{bmatrix}^*
            \Pi(\omega)
            \begin{bmatrix}
                H_1(j\omega)\\
                I
            \end{bmatrix}
            &\succeq 0 \label{eq:qc_1} \\
            \begin{bmatrix}
                I \\ -H_2(j\omega)
            \end{bmatrix}^*
            \Pi(\omega)
            \begin{bmatrix}
                I\\
                -H_2(j\omega)
            \end{bmatrix}
            &\prec 0 \label{eq:qc_2}
        \end{align}
    \end{subequations}
\end{thm}

The original gain and phase conditions in Theorem~\ref{thm:mixed_small_gain_phase} can be recovered directly from Theorem~\ref{thm:IQC_theorem} by selecting specific structural forms for the multiplier $\Pi(\omega)$~\cite{ringh2025gain}:
\begin{itemize}
    \item Gain Condition: To recover \eqref{eq:gen_gain_cond}, the multiplier matrix is selected as
    \begin{equation} \label{eq:gain_pi}
        \Pi(\omega) = \begin{bmatrix}
            -\gamma^2(\omega)I & 0\\
            0 & I
        \end{bmatrix},
    \end{equation}
    where $\gamma(\omega)> 0$ serves as a slack variable parameterizing the frequency-dependent gain bound.
    
    \item Phase Condition: To recover \eqref{eq:gen_phase_cond}, the multiplier matrix is structured as
    \begin{equation} \label{eq:phase_pi}
        \Pi(\omega) = \begin{bmatrix} 
            0 & z(\omega)I\\
            z^*(\omega)I & 0
        \end{bmatrix},
    \end{equation}
    where $z(\omega) \in \mathbb{C}$ is a complex slack variable encoding the frequency-dependent phase bounds.
\end{itemize}
\medskip
While the structural choices in \eqref{eq:gain_pi} and \eqref{eq:phase_pi} map directly to gain and phase bounds respectively, they represent only a particular choice of the valid multipliers. In reality, $\Pi(\omega)$ offers a broader and more flexible search space. This degree of freedom will be exploited in Section \ref{sec:4_ConverterWise} to provide the partitioned stability conditions.

%% file: sec3_DecConditions.tex
\section{Decentralized Stability Criteria for Power Grid Interconnections}
\begin{figure}
    \centering
    \includegraphics[width=\linewidth]{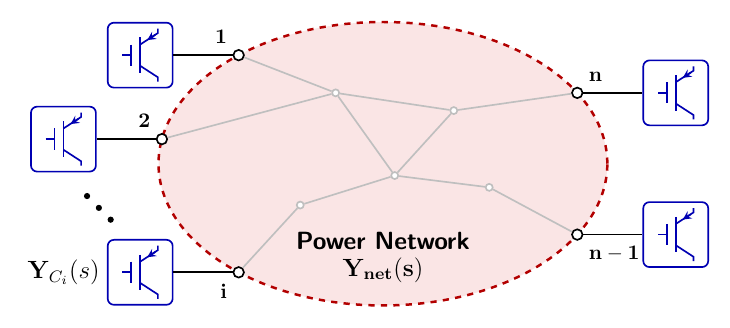}
    \caption{Schematic representation of a multi-converter power system}
    \label{fig:multi_converter_diagram}
\end{figure}

This paper focuses on evaluating the small-signal stability of an interconnected power grid network comprising $n$ grid-connected power electronic converters, as depicted in Figure \ref{fig:multi_converter_diagram}. 
We define the collective converter dynamics using a block-diagonal admittance matrix:
\begin{equation}
    \mathbf{Y}_C(s)=\text{diag}(Y_{C_1}(s),Y_{C_2}(s),\cdots,Y_{C_n}(s))
    \label{eq:block_diagonal}
\end{equation}
where $Y_{C_i}(s)$ captures the individual admittance characteristic of the $i$-th converter. The transmission grid is modeled via the network admittance matrix $\mathbf{Y}_{net}(s)$. We utilize a Kron-reduction technique on $\mathbf{Y}_{net}(s)$ to eliminate intermediate internal nodes. All transfer functions and admittances are derived with respect to a unified, global $dq$ reference frame. Paralleling the structure of Figure \ref{fig:fd_diagram}, this integrated multi-converter network can be modeled as a closed-loop feedback system, as illustrated in Figure \ref{fig:fd_diagram_conv_net}.

Throughout this paper, we employ the frequency-dependent frame transformation introduced in~\cite{cifelli2025decentralized}. The motivation is that converter admittances exhibit distinct behavior at low and high frequencies. At high frequencies, the dynamics are dominated by the output filter and the inner current and voltage control loops, which typically exhibit favorable phase properties and can be analyzed directly in the original admittance representation. At low frequencies, instead, the dynamics are governed by the outer P-f and Q-V control loops, which generally do not satisfy the sectoriality conditions required for applying phase-based criteria. These properties are, however, typically recovered in the Voltage-Angle to P-Q frame. The transformation therefore provides a smooth interpolation between these two representations, enabling consistent analysis across the frequency range, while an additional weighting function can be used to reduce conservativeness.

Accordingly, the transformed converter and network admittances are defined as
\begin{equation}
    J_{C_i} = (\mathcal{E}_i(s)Y_{C_i}(s)+\mathcal{C}_i(s))\mathcal{F}_i(s)
    \label{eq:conv_frame_transf}
\end{equation}
\begin{equation}
    \mathbf{J}_{net} = (\boldsymbol{\mathcal{E}}(s)\mathbf{Y}_{net}(s)-\boldsymbol{\mathcal{C}}(s))\boldsymbol{\mathcal{F}}(s)
    \label{eq:net_frame_transf}
\end{equation}
where the frequency-dependent transformation matrices are given by
\begin{equation}
\begin{aligned}
    &\mathcal{E}_i(s) = \mathcal{E}_{J_i} \\
    &\mathcal{C}_i(s) = H_{\mathrm{LPF}}(s)\mathcal{C}_{J_i}\\
    &\mathcal{F}_i(s) = H_{\mathrm{LPF}}(s)\mathcal{F}_{J_i}+H_{\mathrm{HPF}}(s)W_i(s)\mathcal{E}_{J_i}^{-1}
\end{aligned}
\label{eq:freq_dep_tf_better}
\end{equation}
The matrices $\mathcal{E}_{J_i}$, $\mathcal{F}_{J_i}$, and $\mathcal{C}_{J_i}$ depend on the operating point and are reported in~\cite{cifelli2025decentralized}. Here, $H_{\mathrm{LPF}}(s)$ and $H_{\mathrm{HPF}}(s)$ denote complementary low and high-pass filters, respectively, while $W_i(s)$ is an optional weighting function that adds flexibility and reduces conservativeness. For instance, in~\cite{cifelli2025decentralized} it is chosen as the inverse virtual admittance, whereas in~\cite{huangGainPhaseDecentralized2024} it is taken as the average network line impedance.

\color{black}
We are now ready to apply Theorem \ref{thm:mixed_small_gain_phase} to the system in Figure \ref{fig:fd_diagram_conv_net}, considering the transformed systems $J_{C_i}$ and $\mathbf{J}_{net}$. As shown in~\cite{cifelli2025decentralized}, the transformation \eqref{eq:conv_frame_transf}--\eqref{eq:net_frame_transf} preserves closed-loop stability, so that stability of the transformed feedback interconnection implies stability of the original multi-converter system.
The block-diagonal structure of $\mathbf{Y}_C(s)$, which is preserved under the transformation, allows a decentralized stability criterion, leading to the following result \cite{huangGainPhaseDecentralized2024}. 

\begin{thm}[Decentralized Mixed Small Gain-Phase Theorem]
    The multi-converter system in Figure \ref{fig:fd_diagram_conv_net} is stable if the $J_{C_i}$ and $\mathbf{J}_{net}^{-1}$ are stable, and for each $\omega \in [0,\infty]$, either
    \begin{enumerate}
        \item 
        the decentralized gain condition, i.e.,
        \begin{equation}
            \overline{\sigma}(J_{C_i}(j\omega))<\underline{\sigma}(\mathbf{J}_{net}(j\omega)) \quad \forall i
            \label{eq:uniform_gain_cond}
        \end{equation}
        holds, or
        \item 
        the decentralized phase condition, i.e.,
        every $J_{C_i}(j\omega)$ and $\mathbf{J}^{-1}_{net}(j\omega)$ are sectorial, and
        \begin{subequations}
        \begin{equation}
            \overline{\phi}(J_{C_i}(j\omega))<\pi-\overline{\phi}(\mathbf{J}_{net}^{-1}(j\omega)) \quad \forall i
        \end{equation}
        \begin{equation}
            \underline{\phi}(J_{C_i}(j\omega))>-\pi-\underline{\phi}(\mathbf{J}_{net}^{-1}(j\omega))\quad \forall i
        \end{equation}
        \begin{equation}
            \max_i\overline{\phi}(J_{C_i}(j\omega)) -\min_i\underline{\phi}(J_{C_i}(j\omega))<\pi
        \end{equation}
        \label{eq:uniform_phase_cond}
        \end{subequations}
    \end{enumerate}
    \label{thm:dec_mixed_small_gain_phase}
\end{thm}
This decentralized formulation allows the stability of a large-scale multi-converter system to be certified by verifying local conditions on each converter model together with the network model, without requiring the explicit computation of the full interconnected system’s dynamics.

The practical applicability of Theorem \ref{thm:dec_mixed_small_gain_phase} is inherently limited by the requirement that all converters in a system must satisfy the same condition, either gain or phase.
To illustrate this limitation, consider the two-converter system in Figure \ref{fig:2ConvSys}, comprising a GFL and a GFM converter connected to an infinite bus. Applying Theorem \ref{thm:dec_mixed_small_gain_phase} to this configuration yields the results shown in Figure \ref{fig:GFM_gain_phase_test}.
At low frequencies, these two specific converters exhibit complementary characteristics that expose the theorem's inflexibility. The GFM converter features high gain (violating the gain condition) but favorable phase properties (satisfying the phase condition). Conversely, the GFL converter features low gain (satisfying the gain condition) but poor phase properties (violating the phase condition).
Because Theorem \ref{thm:dec_mixed_small_gain_phase} enforces a uniform constraint across the entire system, it cannot evaluate the GFL via the gain condition and the GFM via the phase condition simultaneously. 

To resolve this, this paper extends Theorem \ref{thm:dec_mixed_small_gain_phase} to allow different converters to satisfy different conditions at each frequency, forming the main contribution detailed in the next sections.

\color{black}

\begin{figure}
    \centering
    \includegraphics[width=\linewidth]{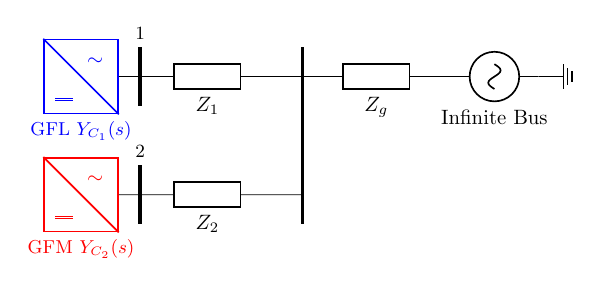}
    \caption{Two-Converter System}
    \label{fig:2ConvSys}
\end{figure}

\begin{figure}
    \centering
    \begin{subfigure}{\columnwidth}
        \centering
        \includegraphics[width=\linewidth]{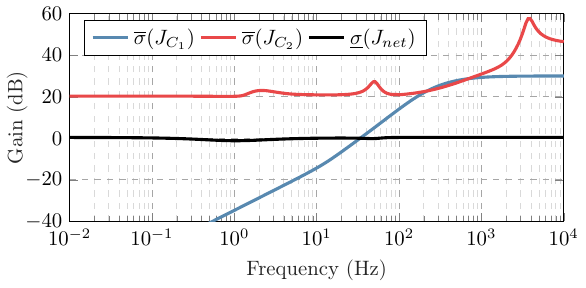}
    \end{subfigure}
    \begin{subfigure}{\columnwidth}
        \centering
        \includegraphics[width=\linewidth]{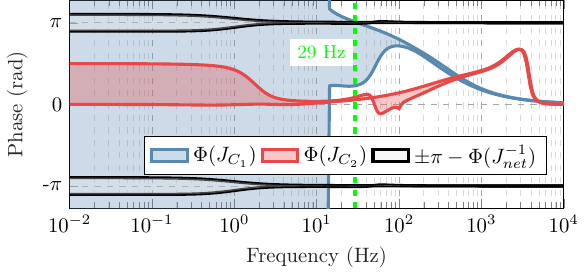}
    \end{subfigure}
    \caption{Mixed small gain–phase condition for a two-converter system. If a system is not sectorial, its phase is represented by the interval $[-2\pi,\, 2\pi]$. The small-gain condition is violated if the converter gain exceeds the network gain. The small-phase condition is violated if the converter phase overlaps the network phase shifted by $\pm\pi$.}
    \label{fig:GFM_gain_phase_test}
\end{figure}

%% file: sec4_ConverterWise.tex
\section{Partitioned Mixed Gain--Phase Stability Conditions}
\label{sec:4_ConverterWise}
As mentioned previously, it would be beneficial to have a more flexible decentralized criterion that allows different converters to satisfy different conditions at each frequency. In this regard, consider partitioning the set of $n$ converters into two subsets $\mathcal{I}_g$ of $m$ converters and $\mathcal{I}_\phi$ of $l=n-m$ converters. Without loss of generality, we assume the system indices are reordered such that the first $m$ ($1$ to $m$) buses and their corresponding converters belong to $\mathcal{I}_g$, while the remaining $l$ ($m+1$ to $n$) belong to $\mathcal{I}_\phi$.
Then, we can extend Theorem \ref{thm:dec_mixed_small_gain_phase} as follows.
\begin{thm}[Partitioned Mixed Small Gain--Phase Theorem]
    The multi-converter system in Figure \ref{fig:fd_diagram_conv_net} is stable if $J_{C_i}$ and $\mathbf{J}_{net}^{-1}$ are stable, and for each $\omega \in [0,\infty]$, either
    \begin{enumerate}
        \item 
        the decentralized gain condition \eqref{eq:uniform_gain_cond} holds, or
        \item 
        the decentralized phase condition \eqref{eq:uniform_phase_cond} holds, or
        \item each $J_{C_i}(j\omega)$ is sectorial and the partitioned decentralized gain--phase condition holds, i.e.,
        \begin{subequations}
        \begin{equation}
            \overline{\sigma}(J_{C_i}(j\omega))<\gamma(\omega) \quad \forall i \in \mathcal{I}_g
            \label{eq:gain_condition_gain_converters}
        \end{equation}
        \begin{equation}
            \overline{\phi}(J_{C_i}(j\omega))<\frac{\pi}{2}-\theta(\omega) \quad \forall i \in \mathcal{I}_\phi
            \label{eq:phase_upper_bound}
        \end{equation}
        \begin{equation}
            \underline{\phi}(J_{C_i}(j\omega))>-\frac{\pi}{2}-\theta(\omega)\quad \forall i \in \mathcal{I}_\phi
            \label{eq:phase_lower_bound}
        \end{equation}
        \label{eq:ext_thm_conds}
        \end{subequations}
        where $\gamma(\omega)>0$ and $\theta(\omega) \in [-\pi,\pi]$ are a solution of the following quadratic constraint:
        \begin{equation}
            \begin{bmatrix}
                \mathbf{J}_{net}^{-1}(j\omega) \\ I_n
            \end{bmatrix}^*
            \Pi(\omega)
            \begin{bmatrix}
                \mathbf{J}_{net}^{-1}(j\omega) \\ I_n
            \end{bmatrix}
            \succeq 0
            \label{eq:Jnet_QC}
        \end{equation}
        with
        \begin{equation}
            \Pi(\omega) = \begin{bmatrix}
                -\gamma I_m & 0 & 0 & 0 \\
                0 & 0 & 0 & \kappa e^{j\theta}I_l \\
                0 & 0 & \gamma^{-1}I_m & 0 \\
                0 & \kappa e^{-j\theta}I_l & 0 & 0
            \end{bmatrix}
            \label{eq:Pi_phase_gain}
        \end{equation}
        and $\kappa(\omega)>0$.
    \end{enumerate}  
    \label{thm:ext_mixed_small_gain_phase}
\end{thm}
\begin{proof}
    The proof is based on the application of Theorem \ref{thm:IQC_theorem} using a multiplier associated with each condition. Let $H_1(j\omega)=\mathbf{J}_{net}^{-1}(j\omega)$ and $H_2(j\omega)=\mathbf{J}_{C}(j\omega)$.
    Should condition 1 be satisfied, there exists a $\Pi(\omega)$ of the form in \eqref{eq:gain_pi} such that the QCs in \eqref{eq:qc} are satisfied.
    Should condition 2 be satisfied, there exists a $\Pi(\omega)$ of the form in \eqref{eq:phase_pi} such that the QCs in \eqref{eq:qc} are satisfied.
    In the case of condition 3, the QC related to the network \eqref{eq:qc_1} is already satisfied  by \eqref{eq:Jnet_QC}. 
    It remains to be shown that if condition 3 holds, then $\Pi$ satisfies the QCs related to the converters in \eqref{eq:qc_2}, i.e.:
    \begin{equation*}
            \begin{bmatrix}
                I \\ -\mathbf{J}_{C}(j\omega)
            \end{bmatrix}^*
            \Pi(\omega)
            \begin{bmatrix}
                I \\ -\mathbf{J}_{C}(j\omega)
            \end{bmatrix}
            \prec 0
        \end{equation*}
    Considering $\mathbf{J}_{C}$ partitioned as $\text{diag}(\mathbf{J}^g_{C},\mathbf{J}^{\phi}_{C})$, where $\mathbf{J}^g_{C}$ contains the ${J}_{C_i}$ for $i \in \mathcal{I}_g$ and $\mathbf{J}^{\phi}_{C}$ contains the ${J}_{C_i}$ for $i \in \mathcal{I}_\phi$, we can rewrite the above matrix inequality as two separate inequalities as follows:
    \begin{equation}
            -\gamma^2I + (\mathbf{J}^g_{C})^*\mathbf{J}^g_{C} \prec 0 \\
            \label{eq:proof_dec_gain}
    \end{equation} 
    \begin{equation}
            -\kappa e^{j\theta}\mathbf{J}^{\phi}_{C}-\kappa e^{-j\theta}(\mathbf{J}^{\phi}_{C})^* \prec 0
            \label{eq:proof_dec_phase}
    \end{equation} 
    Since $\mathbf{J}^g_{C}$ is block diagonal, inequality \eqref{eq:proof_dec_gain} is guaranteed to be satisfied by \eqref{eq:gain_condition_gain_converters}. While inequality \eqref{eq:proof_dec_phase} asserts that rotating the numerical range of $\mathbf{J}^{\phi}_{C}$  by $\theta$ places it in the positive real half plane, or equivalently:
    \begin{equation*}
        \overline{\phi}(\mathbf{J}^{\phi}_{C})+\theta<\frac{\pi}{2}, \quad 
        \underline{\phi}(\mathbf{J}^{\phi}_{C})+\theta>-\frac{\pi}{2}
    \end{equation*}
    Again, by block diagonality of $\mathbf{J}^{\phi}_{C}$, these conditions are satisfied by \eqref{eq:phase_upper_bound}-\eqref{eq:phase_lower_bound}.
    Therefore, for each $\omega$, if one of the three conditions is satisfied, then there exists a $\Pi(\omega)$ such that the QCs in \eqref{eq:qc} are satisfied. Then, by Theorem \ref{thm:IQC_theorem}, the closed loop system is stable.
\end{proof}
The extension lies in the third condition, where stability is guaranteed if a chosen subset of converters satisfies a gain condition, while the complementary subset satisfies a phase condition. Notably, the conditions imposed on the individual converters retain exactly the same form as the gain and phase conditions of the uniform theorem; only the bounds differ. 
The key element enabling this partitioning lies in the structure of the multiplier \eqref{eq:Pi_phase_gain}: by assigning a gain-type block to the converters in $\mathcal{I}_g$ and a phase-type block to those in $\mathcal{I}_\phi$, a single multiplier simultaneously encodes both criteria, allowing the two subsets to be treated under different conditions.

In general, inequality \eqref{eq:Jnet_QC} defines the feasible space for the parameters $\gamma$, $\theta$, and $\kappa$.
While the parameter $\kappa$ does not play a direct role in the converters' conditions \eqref{eq:ext_thm_conds}, it shapes the set of valid pairs $(\gamma, \theta)$ that satisfy \eqref{eq:Jnet_QC}. 
Given a specific $\omega$ and $\kappa$, we define the set of pairs $(\gamma, \theta)$ that solve \eqref{eq:Jnet_QC} as:
\begin{equation}
\mathcal{L}_{\kappa}(\omega) = \{(\gamma,\theta) : (\gamma,\kappa,\theta) \text{ satisfies } \eqref{eq:Jnet_QC} \}
\label{eq:set_L_kappa}
\end{equation}


While an infinite number of gain and phase combinations exist, the conditions in \eqref{eq:ext_thm_conds} only need to be satisfied by a single pair of values, offering a degree of design flexibility. Under a centralized approach, the network operator possesses both the network model and all individual converter models. This allows the operator to compute the set $\mathcal{L}_\kappa$ for the network and directly verify whether a common pair exists within the set that satisfies the conditions for all converters simultaneously. Conversely, in a decentralized approach, the network operator has access to the network model but lacks the specific converter models. In this case, a pair must be defined a priori, and an effective trade-off between the gain and phase requirements needs to be found to establish standardized compliance conditions for manufacturers. 
While this paper provides some design guidelines, establishing a truly optimal selection procedure will ultimately rely on practical experience.

\begin{rem}
    In this paper, motivated by the need for a gain condition on GFLs and a phase condition on GFMs, we restrict attention to a two-set partition mixing only gain and phase criteria. More generally, however, by constructing an analogous partitioned multiplier, Theorem~\ref{thm:ext_mixed_small_gain_phase} can be extended to any number of subsets, up to the number of devices, with each subset assigned a different gain, phase, or even passivity-index \cite{chenUnifiedFlexibleFrequencyDomain2025} requirement.
\end{rem}

\subsection{Characterization of $\mathcal{L}_\kappa(\omega)$}
To better understand how a possible pair $(\gamma,\theta)$ in Theorem \ref{thm:ext_mixed_small_gain_phase} can be selected, we study some properties of $\mathcal{L}_\kappa(\omega)$. In particular, we give some results on boundedness and convexity. To this end, we partition the network system $\mathbf{J}_{net}^{-1}$, according to the set $\mathcal{I}_g$ and $\mathcal{I}_{\phi}$ as:
\begin{equation}
    \mathbf{J}_{net}^{-1} = \begin{bmatrix}
        \mathbf{H}_{net}^{gg} & \mathbf{H}_{net}^{g\phi} \\
        \mathbf{H}_{net}^{\phi g} & \mathbf{H}_{net}^{\phi \phi}
    \end{bmatrix}
\end{equation}
and rewrite \eqref{eq:Jnet_QC} as:
\begin{multline}
    \begin{bmatrix}
        -\gamma[\mathbf{H}_{net}^{gg}]^*\mathbf{H}_{net}^{gg}+\gamma^{-1}I_m & -\gamma[\mathbf{H}_{net}^{gg}]^*\mathbf{H}_{net}^{g\phi} \\
        -\gamma[\mathbf{H}_{net}^{g \phi}]^*\mathbf{H}_{net}^{gg} & -\gamma[\mathbf{H}_{net}^{g \phi}]^*\mathbf{H}_{net}^{g \phi}
    \end{bmatrix} + \\
    \kappa \begin{bmatrix}
        0_m & e^{j\theta} [\mathbf{H}_{net}^{\phi g}]^* \\
         e^{-j\theta} \mathbf{H}_{net}^{\phi g} & 
        e^{j\theta} [\mathbf{H}_{net}^{\phi \phi}]^*+ e^{-j\theta} \mathbf{H}_{net}^{\phi \phi} 
    \end{bmatrix} \succeq 0
    \label{eq:Jnet_QC_dec}
\end{multline}
where we divide the dependence on $\gamma$ from $\kappa$ and $\theta$.
The following results provide guarantees on the convexity of the set.
\begin{prop}[Convexity in $\gamma$]
    If $(\tilde{\gamma} ,\theta) \in \mathcal{L}_\kappa(\omega)$, then $(\gamma,\theta) \in \mathcal{L}_\kappa(\omega)$ for all $0 < \gamma \leq \tilde{\gamma}$.
    \label{prop:gamma_conv}
\end{prop}
\begin{proof}
    Given in Appendix \ref{prof:gamma_conv}.
\end{proof}

\begin{prop}[Convexity in $\theta$]
    If $(\gamma,\theta_1), (\gamma,\theta_2) \in \mathcal{L}_\kappa(\omega)$, with $\theta_1 \leq \theta_2$, then $(\gamma,\theta) \in \mathcal{L}_\kappa(\omega)$ for all $\theta_1 \leq \theta \leq \theta_2$.
    \label{prop:theta_conv}
\end{prop}
\begin{proof}
    Given in Appendix \ref{prof:theta_conv}.
\end{proof}
These properties inform the selection of appropriate $\gamma$ and $\theta$ (as shown later), and ensure that the set is well-behaved and non-degenerate. For instance, they guarantee that for a given $\gamma$, the admissible values of $\theta$ do not form disjoint intervals, which would otherwise complicate the choice of $\theta$. Moreover, they establish a monotonic, non-expanding behavior: as $\gamma$ increases, the set of admissible values for $\theta$ shrinks, and vice versa. We can also provide bounds on the set as follows.
\begin{prop}[Bound on maximum $\gamma$]
    If $(\gamma,\theta) \in \mathcal{L}_{\kappa}(\omega)$ then 
    \begin{equation}
        \gamma \leq \underline{\sigma}(\mathbf{H}_{net}^{gg}(j\omega)^{-1})
        \label{eq:gamma_bound}
    \end{equation}
\end{prop}
\begin{proof}
    Inequality \eqref{eq:Jnet_QC_dec} holds only if $-\gamma[\mathbf{H}_{net}^{gg}]^*\mathbf{H}_{net}^{gg}+\gamma^{-1}I_m \succeq 0$, which is equivalent to \eqref{eq:gamma_bound}.
\end{proof}

\begin{prop}[Bound on $\theta$]
    If $(\gamma,\theta) \in \mathcal{L}_{\kappa}(\omega)$ then 
    \begin{equation}
        -\frac{\pi}{2}+\overline{\phi}(\mathbf{H}_{net}^{\phi \phi}(j\omega)) \leq \theta \leq \frac{\pi}{2}+\underline{\phi}(\mathbf{H}_{net}^{\phi \phi}(j\omega))
        \label{eq:theta_bound}
    \end{equation}
\end{prop}
\begin{proof}
    Inequality \eqref{eq:Jnet_QC_dec} holds only if $-\gamma[\mathbf{H}_{net}^{g \phi}]^*\mathbf{H}_{net}^{g \phi}+\kappa e^{j\theta} [\mathbf{H}_{net}^{\phi \phi}]^*+ \kappa e^{-j\theta} \mathbf{H}_{net}^{\phi \phi} \succeq 0$. The first term is negative semi-definite, therefore we must have $e^{j\theta} [\mathbf{H}_{net}^{\phi \phi}]^*+e^{-j\theta} \mathbf{H}_{net}^{\phi \phi} \succeq 0$, which is equivalently to \eqref{eq:theta_bound}.
\end{proof}


These results show that the partitioned criterion can, in general, yield larger limits on gain and phase than the uniform decentralized conditions.
Formally, $\underline{\sigma}(\mathbf{H}_{net}^{gg}(j\omega)^{-1}) \geq \underline{\sigma}(\mathbf{J}_{net}(j\omega))$, $\underline{\phi}(\mathbf{J}_{net}^{-1}(j\omega)) \leq \underline{\phi}(\mathbf{H}_{net}^{\phi \phi}(j\omega))$, and $\overline{\phi}(\mathbf{J}_{net}^{-1}(j\omega)) \geq \overline{\phi}(\mathbf{H}_{net}^{\phi \phi}(j\omega))$~\cite{wangPhasesComplexMatrix2020}. These bounds cannot, however, be reached at the same time: a larger $\gamma$ shrinks the admissible range of $\theta$. The exception is the decoupled case, $\mathbf{H}_{net}^{g\phi}=\mathbf{H}_{net}^{\phi g}=0$, in which the bounds are no longer conservative and describe $\mathcal{L}_\kappa(\omega)$ exactly, so each limit can be attained independently.

\subsection{Example}
\label{sec:example_l_k}
As an illustrative example, we consider the two-converter system shown in Figure \ref{fig:2ConvSys}. Figure \ref{fig:Lkappa_subplots} illustrates how the set $\mathcal{L}_\kappa(\omega)$, for $\kappa=1$, $\omega=100$~Hz, $\mathcal{I}_g=\{1\}$, and $\mathcal{I}_\phi=\{2\}$, varies for different values of $Z_1$ and $Z_2$. The figure also compares the resulting gain and phase limits with those obtained from the uniform small-gain and small-phase criteria, as well as with the bounds given in \eqref{eq:gamma_bound}--\eqref{eq:theta_bound}.

\begin{figure}[!t]
    \centering
    \begin{subfigure}{\linewidth}
        \centering
        \includegraphics[width=\linewidth]{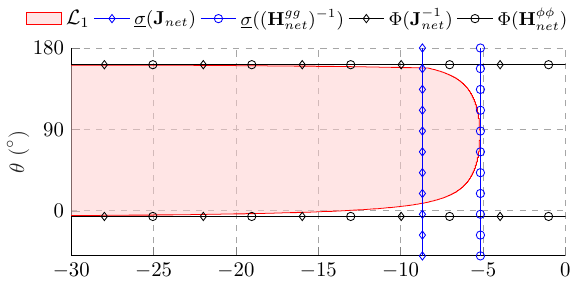}
        \caption{$Z_2 = Z_1$}
        \label{fig:Lkappa_F1}
    \end{subfigure}
    \begin{subfigure}{\linewidth}
        \centering
        \includegraphics[width=\linewidth]{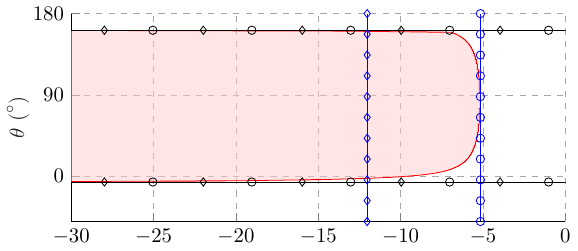}
        \caption{$|Z_2| \gg |Z_1|$}
        \label{fig:Lkappa_F2}
    \end{subfigure}
    \begin{subfigure}{\linewidth}
        \centering
        \includegraphics[width=\linewidth]{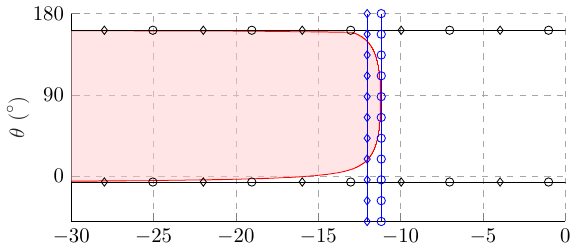}
        \caption{$|Z_2| \ll |Z_1|$}
        \label{fig:Lkappa_F3}
    \end{subfigure}
    \begin{subfigure}{\linewidth}
        \centering
        \includegraphics[width=\linewidth]{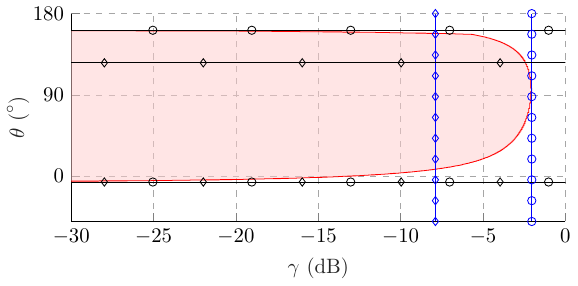}
        \caption{$ R_2/X_2 = X_1/R_1$}
        \label{fig:Lkappa_F4}
    \end{subfigure}
    \caption{Illustration of the $\mathcal{L}_\kappa(\omega)$ and its bounds for the 2 converter system, with different values of $Z_1$ and $Z_2$.}
    \label{fig:Lkappa_subplots}
\end{figure}

We consider four representative cases:
\begin{itemize}
\item Figure \ref{fig:Lkappa_F1} shows the baseline case, in which all line impedances are identical, i.e., $Z_0=Z_1=Z_2=0.1+j0.3$. This case confirms the propositions presented above regarding the shape of $\mathcal{L}_\kappa$. Moreover, the maximum and minimum phase coincide with those obtained from $\mathbf{J}_{net}^{-1}$ owing to the symmetry of the network. The partitioned criterion allows the selection of a larger gain than the corresponding uniform criterion. However, there is a trade-off between gain and phase: higher gain implies stricter choice in phase, and vice versa. This does not mean that there is always a price to pay when partitioning the criteria, since we can ideally obtain better conditions on both gain and phase, e.g. where the two subsets are completely decoupled.

\item Figure \ref{fig:Lkappa_F2} considers the case where the impedance to converter $2$ is significantly larger than that to converter $1$, namely $Z_2=3Z_1$. Here, the proposed theorem not only decomposes the stability requirements into a gain condition for converter $1$ and a phase condition for converter $2$, but also substantially relaxes the gain requirement. The underlying intuition is straightforward. The network contains two connection buses with markedly different strengths. Under the uniform small-gain criterion, a converter is assumed to be connected to either bus, and the resulting condition is dictated by the weakest connection point. 
In contrast, the partitioned criterion applies the gain condition only at the stronger bus and enforces the phase condition at the weaker one. By exploiting the disparity in electrical strength between the two buses, it avoids the conservative gain dictated by the weakest connection, thereby relaxing the gain requirement on converter $1$.


\item Figure \ref{fig:Lkappa_F3} presents the complementary scenario, in which $Z_1$ is much larger than $Z_2$. In this case, only a limited improvement in the gain condition is observed, since the gain criterion is associated with the weaker bus.

\item Finally, Figure \ref{fig:Lkappa_F4} illustrates the effect of non-uniform $R/X$ ratios. Specifically, the ratio $R_1/X_1=1/3$ remains unchanged, while $R_2/X_2=3$. Focusing on the phase condition, the partitioned criterion enlarges the admissible phase region compared with the uniform phase criterion. As in the previous cases with the gain, constraining the converter to a predetermined connection bus reduces the range of phase shifts introduced by the network, thereby expanding the feasible set.
\end{itemize}

In conclusion, these examples provide key insights into when partitioning stability criteria is most advantageous for power system stability analysis, highlighting that its efficacy depends on grid heterogeneity and topology.

When a network features buses with vastly different short-circuit strengths, partitioning enables a strategic segregation of constraints: applying gain conditions to converters at strong buses and phase conditions to those at weak buses safely relaxes the gain requirement where the grid is robust, while enforcing the phase control needed at weak buses, where high gain combined with adverse phase shifts would otherwise trigger instability.

Regarding $R/X$ ratio uniformity, the uniform small-phase criterion is inherently conservative, as it must accommodate the worst-case phase shift across the entire network. Partitioning isolates these variations, reducing the localized phase shift and yielding less restrictive phase bounds. In practice, however, since real-world grids tend to maintain fairly uniform $R/X$ ratios within a voltage tier, the margin gained from ratio non-uniformity may be limited.

Finally, network connectivity largely dictates the overall benefit. In grids that decompose into weakly coupled sub-areas, partitioning yields a substantial relaxation of the conservative uniform bounds, offering considerably more design flexibility. In densely meshed networks, by contrast, the tight coupling limits the advantages of a partitioned analysis.

%% file: sec5_ParSelection.tex
\section{Selection of $\theta$, $\gamma$ and $\kappa$}
\label{sec5}

Condition 3 of Theorem \ref{thm:ext_mixed_small_gain_phase} requires the existence of a triplet $(\gamma,\kappa,\theta)$ satisfying the given inequalities. In theory, one could characterize all the possible $(\gamma,\kappa,\theta)$ that satisfies \eqref{eq:Jnet_QC} and then check if there is any triplet for which the converters' conditions are valid. This can be computation heavy since it needs to be done at each frequency.
In this section we propose a more efficient procedure. Instead of searching over all triplets, we derive bounds for the triplet from the network. Each converter is then checked locally against these bounds, and once all converters satisfy them, the existence of a triplet $(\gamma,\kappa,\theta)$ satisfying~\eqref{eq:Jnet_QC} follows.
First, we assume $\kappa$ is given and we focus on the choice of $(\gamma,\theta)$.

In this scope, we define the minimum and maximum phase $\theta$ given $\gamma$ as follows:
\begin{equation}
    \overline{\theta}(\omega; \gamma) := \max\{\theta: (\gamma,\theta) \in \mathcal{L}_\kappa(\omega)\}
\end{equation}
\begin{equation}
    \underline{\theta}(\omega; \gamma) := \min\{\theta: (\gamma,\theta) \in \mathcal{L}_\kappa(\omega)\}
\end{equation}
These conditions are well-posed given convexity and boundness in $\theta$ of the set $\mathcal{L}_\kappa$. The following result allows to reframe conditions \eqref{eq:phase_upper_bound}-\eqref{eq:phase_lower_bound} in terms of $\overline{\theta}$ and $\underline{\theta}$.

\begin{prop}
    Let $\gamma$ and $\kappa$ be fixed. There exists $\theta \in [\underline{\theta}(\omega;\gamma), \overline{\theta}(\omega;\gamma)]$ such that conditions \eqref{eq:phase_upper_bound}--\eqref{eq:phase_lower_bound} hold if and only if:
    \begin{subequations}
            \begin{equation}
                \overline{\phi}(J_{C_i}(j\omega))<\frac{\pi}{2}-\underline{\theta}(\omega; \gamma) \quad \forall i \in \mathcal{I}_\phi
                \label{eq:thA:2_1}
            \end{equation}
            \begin{equation}
                \underline{\phi}(J_{C_i}(j\omega))>-\frac{\pi}{2}-\overline{\theta}(\omega; \gamma)\quad \forall i \in \mathcal{I}_\phi
                \label{eq:thA:2_2}
            \end{equation}
            \begin{equation}
                \max_{i \in \mathcal{I}_\phi} \overline{\phi}(J_{C_i}(j\omega)) - \min_{i \in \mathcal{I}_\phi} \underline{\phi}(J_{C_i}(j\omega)) < \pi
                \label{eq:thA:2_3}
            \end{equation}
            \label{eq:conditions_with_bounds}
        \end{subequations}
        \label{prop:conditions_with_bounds}
\end{prop}
\begin{proof}
    Given in Appendix \ref{prof:conditions_with_bounds}.
\end{proof} 

Using \eqref{eq:conditions_with_bounds} there is no need to search for a $\theta$, but the condition only needs to be checked on the limits, which is much more practical.

While we ideally seek to maximize the gain $\gamma$, the optimal choices for the phase limits $\overline{\theta}$ and $\underline{\theta}$ are less straightforward, since it depends on the phases of all the converters (which are not known in a decentralized scheme). Therefore, we first select $\gamma$, after which the corresponding phase limits are determined automatically.
Let us define the maximum $\gamma$ across the full set $\mathcal{L}_\kappa$ as:
\begin{equation}
    \overline{\gamma}(\omega) := \max\{\gamma: (\gamma,\theta) \in \mathcal{L}_\kappa(\omega)\}
\end{equation}
Again, the definition is well-posed given the properties of $\mathcal{L}_\kappa$. Taking this as our $\gamma$ in \eqref{eq:gain_condition_gain_converters} would provide very strict limits on the phases. Therefore, a wiser choice is to take $\gamma=\eta\overline{\gamma}$, where $\eta$ is a positive value smaller than one. The value of $\eta$ dictates the tradeoff between gain and phase requirements.

Finally, we address the choice of $\kappa$. Ideally, we seek a unique optimal choice, which is guaranteed only if the family $\{\mathcal{L}_\kappa\}$ admits a greatest element under set inclusion, i.e., a set that contains all others:
\begin{equation*}
\exists\, \tilde{\kappa} \ \text{such that} \ \forall \kappa \ge 0,\ \mathcal{L}_\kappa \subseteq \mathcal{L}_{\tilde{\kappa}}.
\end{equation*}
Although we do not have a general proof of this property, we assume it holds throughout the paper; it has been verified numerically for the examples considered here. Under this assumption, $\mathcal{L}_{\tilde{\kappa}}$ contains all other sets and therefore yields the maximum attainable value of $\gamma$.
Consequently, a natural approach to computing the optimal $\tilde{\kappa}$ is to find the largest $\gamma$ that satisfies \eqref{eq:Jnet_QC} over all admissible $\kappa$ and $\theta$. In its current formulation, the optimization problem is computationally inefficient. However, it can be cast as semi-definite program (SDP) as follows
\begin{equation}
\begin{aligned}
  (\tilde{z},\tilde{c}) = \arg \max_{c,z} \; & c\\
    \textrm{s.t.} \; & \begin{bmatrix}
                \mathbf{J}_{net}^{-1}(j\omega) \\ I_n
            \end{bmatrix}^*
            \Pi(j\omega)
            \begin{bmatrix}
                \mathbf{J}_{net}^{-1}(j\omega) \\ I_n
            \end{bmatrix}
            \succeq 0 \\
            &\Pi(j\omega) = \begin{bmatrix}
                -c I_m & 0 & 0 & 0 \\
                0 & 0 & 0 & zI_l \\
                0 & 0 & I_m & 0 \\
                0 & z^* I_l & 0 & 0
            \end{bmatrix} \\
    &c\geq0,\, z \in \mathbb{C}    \\
\end{aligned}
\label{opt:k_calc}
\end{equation}
and taking:
\begin{equation}
    \kappa = \frac{|\tilde{z}|}{\sqrt{\tilde{c}}}
\end{equation}
The complete procedure is summarized in Algorithm~\ref{alg:kappa_gamma_theta}. The multiplier scaling $\kappa$ and the maximum gain bound $\overline{\gamma}$ are obtained directly from \eqref{opt:k_calc}. The operating gain $\gamma(\omega)$ is then fixed to a fraction $\eta$ of this maximum (representing the trade-off between gain and phase bounds, see Fig.~\ref{fig:Lkappa_subplots}), and the corresponding phase bounds $\underline{\theta}$ and $\overline{\theta}$ are computed (for example, using a bisection search).
\begin{algorithm}
    \caption{Selection of $\kappa$, $\gamma$, and phase bounds for fixed $\omega$}
    \label{alg:kappa_gamma_theta}
    \begin{algorithmic}[1]
        \REQUIRE $\mathbf{J}_{net}^{-1}(j\omega)$, partition sizes $m$ and $l$, scalar $\eta\in(0,1)$
        \STATE Solve \eqref{opt:k_calc} to obtain $(\tilde{z},\tilde{c})$
        \STATE Set $\tilde{\kappa} \leftarrow |\tilde{z}| / \sqrt{\tilde{c}}$
        \STATE Set $\overline{\gamma}(\omega) \leftarrow \sqrt{\tilde{c}}$
        \STATE Choose $\gamma(\omega) \leftarrow \eta\,\overline{\gamma}(\omega)$
        \STATE Compute $\underline{\theta}(\omega;\gamma)$ and $\overline{\theta}(\omega;\gamma)$
        \RETURN $ (\gamma(\omega),\underline{\theta}(\omega;\gamma),\overline{\theta}(\omega;\gamma))$
    \end{algorithmic}
\end{algorithm}

%% file: sec6_StudyCases.tex
\section{Study Cases}
\label{sec:6_StudyCases}

In this section, we present two case studies to demonstrate the application of the proposed method. We first consider a two-converter system and subsequently a 39-bus system to evaluate scalability. The grid-following converters are of a standard type, employing an LC filter, a PI-based control loop, constant-power reference control, and an SRF-PLL for synchronization. The grid-forming converters employ an LCL filter, a virtual synchronous machine (VSM) for power control and synchronization, a virtual admittance, and a resonant current controller. Further details and parameters are available in \cite{Repo}.

\subsection{Two-Converter System}
The system is shown in Fig.~\ref{fig:2ConvSys}. Stability is verified via eigenvalue analysis, with all system eigenvalues in the open left-half plane. Fig.~\ref{fig:GFM_gain_phase_test} illustrates the uniform gain and phase conditions. The phase condition certifies stability above $29$~Hz; below this threshold, no conclusion can be drawn, as neither the small-gain nor the small-phase condition is satisfied simultaneously by both converters, rendering the theorem inconclusive. Nonetheless, the GFM converter satisfies a phase condition while the GFL converter satisfies a gain condition, which is precisely where the partitioned condition in Theorem~\ref{thm:ext_mixed_small_gain_phase} applies. Selecting $\mathcal{I}_g=\{1\}$ for the GFL bus and $\mathcal{I}_\phi=\{2\}$ for the GFM bus exploits this mixed structure to enable certification.
Fig.~\ref{fig:GFM_part_gain_phase_test} depicts the partitioned condition, with gain and phase bounds computed per Algorithm~\ref{alg:kappa_gamma_theta} at each frequency using $\eta=0.8$. The partitioned criterion holds below $36$~Hz, delimited by the point at which the GFL gain exceeds $\gamma$. Since the uniform phase condition holds for $f>29$~Hz and the partitioned condition for $f<36$~Hz, Theorem~\ref{thm:ext_mixed_small_gain_phase} certifies the system as stable.

The comparison against the uniform gain and phase limits (dashed lines) is also shown in Fig.~\ref{fig:GFM_part_gain_phase_test}. Partitioning additionally yields a slightly less restrictive gain condition, since the largest admissible gain attainable under partitioning can exceed the uniform one (e.g. as in Fig.~\ref{fig:Lkappa_F1}). Here the margin is large enough that, even after scaling by $\eta=0.8$, the partitioned gain limit remains above the uniform one. On the phase side, by contrast, the partitioned limits are stricter than in the uniform case; nonetheless, these tighter limits do not compromise the stability certification.

\begin{figure}[tbh]
    \centering
    \begin{subfigure}{\columnwidth}
        \centering
        \includegraphics[width=\linewidth]{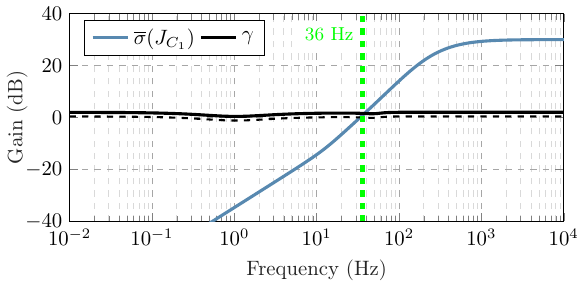}
    \end{subfigure}
    \begin{subfigure}{\columnwidth}
        \centering
        \includegraphics[width=\linewidth]{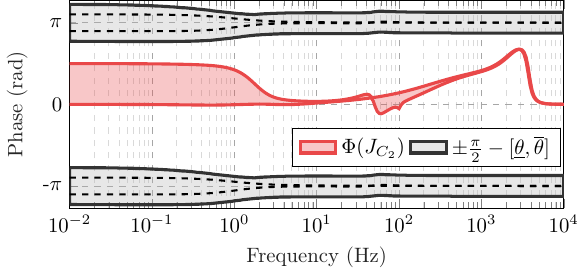}
    \end{subfigure}
    \caption{Partitioned mixed small gain–phase condition for Two-Converter System. Dashed lines represent the network gain and phase limits of the uniform condition, as in Figure \ref{fig:GFM_gain_phase_test}.}
    \label{fig:GFM_part_gain_phase_test}
\end{figure}

\subsection{IEEE 39 Bus System}

\begin{figure}
    \centering
    \includegraphics[width=\linewidth]{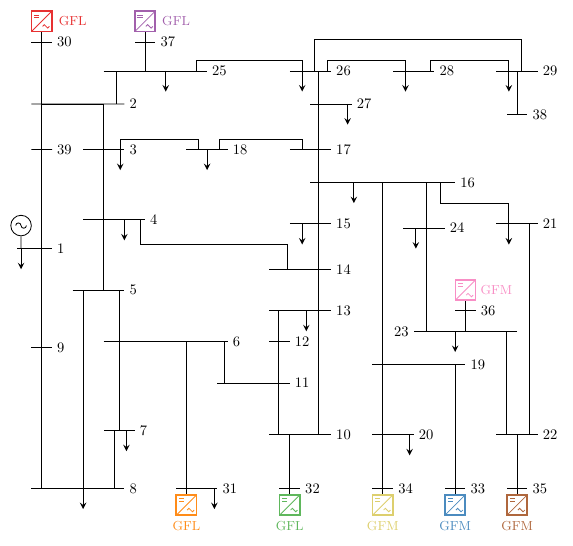}
    \caption{IEEE 39 Bus System}
    \label{fig:39ConvSys}
\end{figure}

As an additional use case, we consider the IEEE 39 Bus system depicted in Figure \ref{fig:39ConvSys}. We consider 4 GFLs and 4 GFMs connected at the locations shown in the figure. All GFMs implement the same control structure, but with different parametrization and operating point \cite{Repo}. Similarly for the GFL converters. 
Figure \ref{fig:39_mixed} shows the uniform mixed decentralized conditions of Theorem \ref{thm:dec_mixed_small_gain_phase}. 
In particular, Figure \ref{fig:39_mixed_gain} displays the gain condition, which is never satisfied since GFMs have higher gain than the network at all frequencies. 
Figure \ref{fig:39_mixed_phase_gfl}, for GFLs, and \ref{fig:39_mixed_phase_gfm}, for GFMs, show the small-phase condition. While GFMs satisfy the condition at all frequencies, GFLs do so only above $60$~Hz. Overall, the small-phase condition is satisfied for $f>60$Hz. 
Therefore, with only Theorem \ref{thm:dec_mixed_small_gain_phase}, we would not be able to certify stability since for $f\le60$Hz, neither condition is satisfied.

In Figure \ref{fig:39_iqc}, the new third condition of the extended Theorem \ref{thm:ext_mixed_small_gain_phase} is shown.
The gain set $\mathcal{I}_g$ is selected as the set of GFLs, i.e. $\mathcal{I}_g=\{30,31,32,37\}$, while the phase set $\mathcal{I}_\phi$ is selected as the set of GFMs, i.e. $\mathcal{I}_\phi=\{33,34,35,36\}$.
The gain limit $\gamma$ and the phase limits $\underline{\theta}$ and $\overline{\theta}$ are computed using Algorithm \ref{alg:kappa_gamma_theta} with $\eta=0.8$.
Figure \ref{fig:39_iqc_gain_gfl} shows that GFLs satisfy the gain condition for $f < 88$ Hz, while \ref{fig:39_iqc_phase_gfm} shows that GFMs satisfy the phase condition at all frequencies. Then, overall the partitioned condition is satisfied for $f<88$~Hz. Therefore, this new condition let us conclude stability of the system, since the small-phase condition holds for $f>60$~Hz and the partitioned gain--phase condition for $f<88$~Hz.
Although we kept $\eta$ constant in this example, it could instead be chosen frequency-dependent; for instance, below $1$~Hz, where the GFLs' gain lies well below the limit, $\eta$ could be decreased to allow for a larger phase on the GFMs.

Figure \ref{fig:39_iqc} shows the comparison with the uniform gain and phase bounds, in dashed lines. Up to $1$~kHz, the partitioned gain condition is less restrictive compared to the uniform one, while the phase limits are stricter than the uniform ones. At higher frequencies, the partitioned bounds are similar to the uniform ones. In general, there is no rule establishing whether the partitioned gain and phase limits are less or more restrictive than the uniform ones; as explained in Section \ref{sec:example_l_k}, this strongly depends on the network characteristics.

\begin{figure}
    \centering
    \begin{subfigure}{\linewidth}
        \centering
        \includegraphics[width=\linewidth]{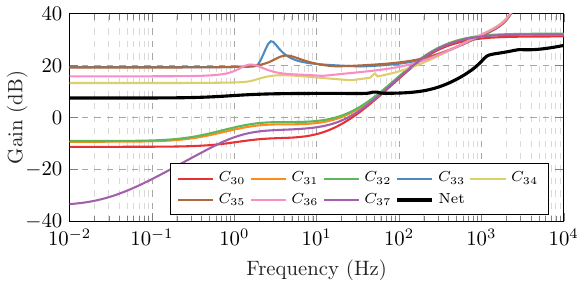}
        \caption{Uniform Small-Gain}
        \label{fig:39_mixed_gain}
    \end{subfigure}
    \begin{subfigure}{\linewidth}
        \centering
        \includegraphics[width=\linewidth]{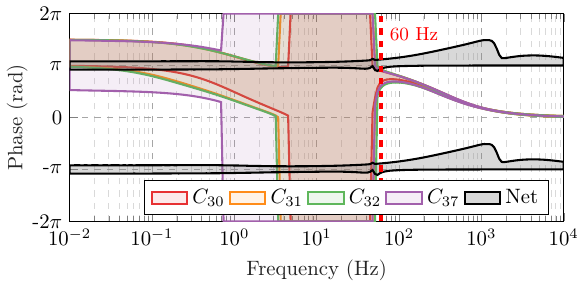}
        \caption{Uniform Small-Phase (Only GFLs)}
        \label{fig:39_mixed_phase_gfl}
    \end{subfigure}
    \begin{subfigure}{\linewidth}
        \centering
        \includegraphics[width=\linewidth]{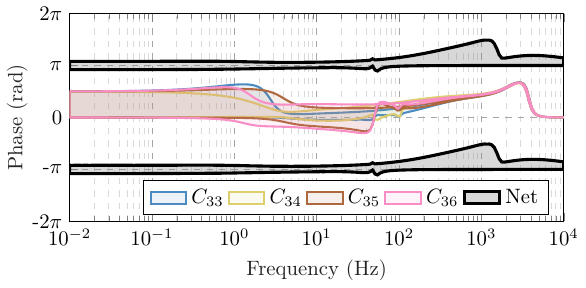}
        \caption{Uniform Small-Phase (Only GFMs)}
        \label{fig:39_mixed_phase_gfm}
    \end{subfigure}
    \caption{IEEE 39 Bus System: Uniform Small Gain--Phase Conditions.}
    \label{fig:39_mixed}
\end{figure}

\begin{figure}
    \centering
    \begin{subfigure}{\linewidth}
        \centering
        \includegraphics[width=\linewidth]{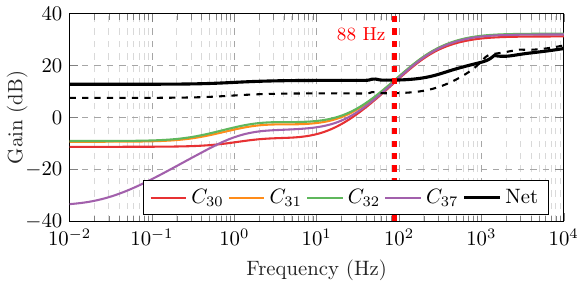}
        \caption{Partitioned Conditions: Small-Gain (Only GFLs)}
        \label{fig:39_iqc_gain_gfl}
    \end{subfigure}
    \begin{subfigure}{\linewidth}
        \centering
        \includegraphics[width=\linewidth]{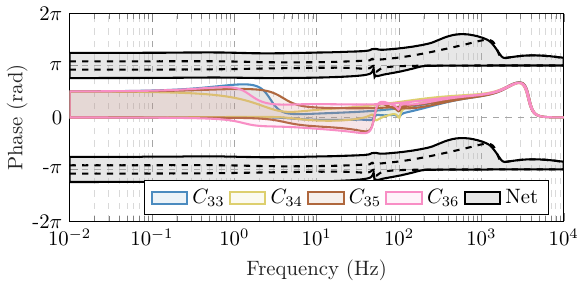}
        \caption{Partitioned Conditions: Small-Phase (Only GFMs)}
        \label{fig:39_iqc_phase_gfm}
    \end{subfigure}
    \caption{IEEE 39 Bus System: Partitioned Small Gain--Phase Conditions. Dashed lines represent the network gain and phase limits of the uniform condition, as in Figure \ref{fig:39_mixed}.} 
    \label{fig:39_iqc}
\end{figure}

%% file: secN_Conclusion.tex
\section{Conclusion}
\label{sec:conc}
This paper introduced a partitioned mixed gain–phase decentralized stability criterion. Rather than forcing every converter to satisfy the same type of condition at each frequency, it partitions them into subsets certified simultaneously through different requirements, a small-gain bound for one subset and a small-phase bound for the other, with the trade-off governed by a network-dependent quadratic constraint.
The method certified stability on a two-converter and the IEEE 39-bus system, where the uniform conditions were previously inconclusive.

The central value of the approach is that it turns a uniform test into a technology-aware certificate that remains fully local at the converter level. By assigning gain conditions to grid-following units and phase conditions to grid-forming units, it exploits the complementary characteristics of the two technologies at low frequencies, the low gain of GFLs and the favorable phase of GFMs, that a uniform test cannot leverage. Because each converter is also tied to a specific bus, the partition additionally exploits network location awareness as a byproduct: it can relax the gain limit at electrically strong buses while enforcing phase at weak ones, though the size of this margin depends strongly on grid heterogeneity and topology.

The criterion naturally extends beyond the gain/phase dichotomy: since each
requirement is encoded in a separate block of the multiplier, passivity indices can be incorporated on the same footing, assigning to each subset the condition that best matches its characteristics. In practice, we envision two complementary uses.
In a centralized setting, the partitioned criterion is applied as a second stage wherever the uniform conditions fail, grouping the converters according to the condition each of them satisfies. In a decentralized setting, the operator selects the type of requirement (gain, phase, or even passivity index) for each bus, based on the expected converter technology and on operational experience, computes the corresponding limits from the network model, and provides them to manufacturers as local compliance conditions, verifiable without knowledge of the rest of the system.